\documentclass[aps,notitlepage,twocolumn,superscriptaddress]{revtex4-2}

\makeatletter
\def\@hangfrom@section#1#2#3{\normalsize\@hangfrom{#1#2}#3}
\def\@hangfroms@section#1#2{\normalsize#1#2}
\makeatother

\usepackage{amsmath,amsfonts,amssymb,bm}
\usepackage{graphicx,color}
\usepackage{enumerate}
\usepackage{diagbox}
\usepackage{makecell}
\usepackage[colorlinks=true, allcolors=blue]{hyperref}

\usepackage[dvipsnames]{xcolor}
\usepackage{framed}
\definecolor{shadecolor}{rgb}{0.9,0.9,0.9}

\usepackage{mathtools}
\usepackage{amsmath}
\usepackage[shortlabels]{enumitem}

\usepackage{graphicx,epic,eepic,epsfig,amsmath,latexsym,amssymb,verbatim,color}
 
\usepackage{amsfonts}       
\usepackage{nicefrac}       

\usepackage{amsmath}
\usepackage{bbm}
 
\usepackage{diagbox}
\usepackage{float}
\usepackage{tikz}
\usetikzlibrary{chains}
\usetikzlibrary{fit}
\usepackage{pgflibraryarrows}		
\usepackage{pgflibrarysnakes}		

\usepackage{epsfig}
\usetikzlibrary{shapes.symbols,patterns} 
\usepackage{pgfplots}

\usepackage[strict]{changepage}

\usepackage[marginal]{footmisc}
\usepackage{url}
\usepackage{theorem}

\newtheorem{definition}{Definition}
\newtheorem{proposition}{Proposition}
\newtheorem{lemma}[proposition]{Lemma}

\newtheorem{theorem}[proposition]{Theorem}

\newtheorem{corollary}[proposition]{Corollary}

\def\squareforqed{\hbox{\rlap{$\sqcap$}$\sqcup$}}
\def\qed{\ifmmode\squareforqed\else{\unskip\nobreak\hfil
\penalty50\hskip1em\null\nobreak\hfil\squareforqed
\parfillskip=0pt\finalhyphendemerits=0\endgraf}\fi}
\def\endenv{\ifmmode\;\else{\unskip\nobreak\hfil
\penalty50\hskip1em\null\nobreak\hfil\;
\parfillskip=0pt\finalhyphendemerits=0\endgraf}\fi}
\newenvironment{proof}{\noindent \textbf{{Proof~} }}{\hfill $\blacksquare$}

\newcounter{remark}

\newcounter{example}
\newenvironment{example}[1][]{\refstepcounter{example}\par\medskip\noindent%
\textbf{Example~\theexample #1} }{\medskip}

\mathchardef\ordinarycolon\mathcode`\:
\mathcode`\:=\string"8000
\def\vcentcolon{\mathrel{\mathop\ordinarycolon}}
\begingroup \catcode`\:=\active
  \lowercase{\endgroup
  \let :\vcentcolon
  }

\usepackage{cleveref}
\usepackage{graphicx}
\usepackage{xcolor}

\RequirePackage[framemethod=default]{mdframed}
\newmdenv[skipabove=7pt,
skipbelow=7pt,
backgroundcolor=darkblue!15,
innerleftmargin=5pt,
innerrightmargin=5pt,
innertopmargin=5pt,
leftmargin=0cm,
rightmargin=0cm,
innerbottommargin=5pt,
linewidth=1pt]{tBox}

\newmdenv[skipabove=7pt,
skipbelow=7pt,
backgroundcolor=red!15,
innerleftmargin=5pt,
innerrightmargin=5pt,
innertopmargin=5pt,
leftmargin=0cm,
rightmargin=0cm,
innerbottommargin=5pt,
linewidth=1pt]{rBox}

\newmdenv[skipabove=7pt,
skipbelow=7pt,
backgroundcolor=blue2!25,
innerleftmargin=5pt,
innerrightmargin=5pt,
innertopmargin=5pt,
leftmargin=0cm,
rightmargin=0cm,
innerbottommargin=5pt,
linewidth=1pt]{dBox}
\newmdenv[skipabove=7pt,
skipbelow=7pt,
backgroundcolor=darkkblue!15,
innerleftmargin=5pt,
innerrightmargin=5pt,
innertopmargin=5pt,
leftmargin=0cm,
rightmargin=0cm,
innerbottommargin=5pt,
linewidth=1pt]{sBox}
\definecolor{darkblue}{RGB}{0,76,156}
\definecolor{darkkblue}{RGB}{0,0,153}
\definecolor{blue2}{RGB}{102,178,255}
\definecolor{darkred}{RGB}{195,0,0}

\newcommand{\nc}{\newcommand}
\nc{\rnc}{\renewcommand}
\nc{\lbar}[1]{\overline{#1}}
\nc{\bra}[1]{\langle#1|}
\nc{\ket}[1]{|#1\rangle}
\nc{\ketbra}[2]{|#1\rangle\!\langle#2|}
\nc{\braket}[2]{\langle#1|#2\rangle}

\nc{\proj}[1]{| #1\rangle\!\langle #1 |}
\nc{\avg}[1]{\langle#1\rangle}
\nc{\smfrac}[2]{\mbox{$\frac{#1}{#2}$}}
\nc{\tr}{\operatorname{Tr}}
\nc{\ox}{\otimes}
\nc{\dg}{\dagger}
\nc{\dn}{\downarrow}
\nc{\cA}{{\cal A}}
\nc{\cB}{{\cal B}}
\nc{\cC}{{\cal C}}
\nc{\cD}{{\cal D}}
\nc{\cE}{{\cal E}}
\nc{\cF}{{\cal F}}
\nc{\cG}{{\cal G}}
\nc{\cH}{{\cal H}}
\nc{\cI}{{\cal I}}
\nc{\cJ}{{\cal J}}
\nc{\cK}{{\cal K}}
\nc{\cL}{{\cal L}}
\nc{\cM}{{\cal M}}
\nc{\cN}{{\cal N}}
\nc{\cO}{{\cal O}}
\nc{\cP}{{\cal P}}
\nc{\cQ}{{\cal Q}}
\nc{\cR}{{\cal R}}
\nc{\cS}{{\cal S}}
\nc{\cT}{{\cal T}}
\nc{\cU}{{\cal U}}
\nc{\cV}{{\cal V}}
\nc{\cX}{{\cal X}}
\nc{\cY}{{\cal Y}}
\nc{\cZ}{{\cal Z}}
\nc{\cW}{{\cal W}}
\nc{\csupp}{{\operatorname{csupp}}}
\nc{\qsupp}{{\operatorname{qsupp}}}
\nc{\var}{{\operatorname{var}}}
\nc{\rar}{\rightarrow}
\nc{\lrar}{\longrightarrow}
\nc{\polylog}{{\operatorname{polylog}}}
\nc{\wt}{{\operatorname{wt}}}
\nc{\av}[1]{{\left\langle {#1} \right\rangle}}
\nc{\supp}{{\operatorname{supp}}}

\nc{\argmin}{{\operatorname{argmin}}}

\def\x{\xi}

\nc{\RR}{{{\mathbb R}}}
\nc{\CC}{{{\mathbb C}}}
\nc{\FF}{{{\mathbb F}}}
\nc{\NN}{{{\mathbb N}}}
\nc{\ZZ}{{{\mathbb Z}}}
\nc{\PP}{{{\mathbb P}}}
\nc{\QQ}{{{\mathbb Q}}}
\nc{\UU}{{{\mathbb U}}}
\nc{\EE}{{{\mathbb E}}}
\nc{\id}{{\operatorname{id}}}

\nc{\CHSH}{{\operatorname{CHSH}}}

\nc{\be}{\begin{equation}}
\nc{\ee}{{\end{equation}}}
\nc{\bea}{\begin{eqnarray}}
\nc{\eea}{\end{eqnarray}}
\nc{\<}{\langle}
\rnc{\>}{\rangle}
\nc{\rU}{\mbox{U}}

\nc{\ob}[1]{#1}

\nc{\SEP}{{\text{\rm SEP}}}
\nc{\NS}{{\text{\rm NS}}}
\nc{\LOCC}{{\text{\rm LOCC}}}
\nc{\PPT}{{\text{\rm PPT}}}
\nc{\EXT}{{\text{\rm EXT}}}
\nc{\Sym}{{\operatorname{Sym}}}

\nc{\ERLO}{{E_{\text{r,LO}}}}
\nc{\ERLOCC}{{E_{\text{r,LOCC}}}}
\nc{\ERPPT}{{E_{\text{r,PPT}}}}
\nc{\ERLOCCinfty}{{E^{\infty}_{\text{r,LOCC}}}}
\nc{\Aram}{{\operatorname{\sf A}}}

\makeatletter
\def\grd@save@target#1{%
  \def\grd@target{#1}}
\def\grd@save@start#1{%
  \def\grd@start{#1}}
\tikzset{
  grid with coordinates/.style={
    to path={%
      \pgfextra{%
        \edef\grd@@target{(\tikztotarget)}%
        \tikz@scan@one@point\grd@save@target\grd@@target\relax
        \edef\grd@@start{(\tikztostart)}%
        \tikz@scan@one@point\grd@save@start\grd@@start\relax
        \draw[minor help lines,magenta] (\tikztostart) grid (\tikztotarget);
        \draw[major help lines] (\tikztostart) grid (\tikztotarget);
        \grd@start
        \pgfmathsetmacro{\grd@xa}{\the\pgf@x/1cm}
        \pgfmathsetmacro{\grd@ya}{\the\pgf@y/1cm}
        \grd@target
        \pgfmathsetmacro{\grd@xb}{\the\pgf@x/1cm}
        \pgfmathsetmacro{\grd@yb}{\the\pgf@y/1cm}
        \pgfmathsetmacro{\grd@xc}{\grd@xa + \pgfkeysvalueof{/tikz/grid with coordinates/major step}}
        \pgfmathsetmacro{\grd@yc}{\grd@ya + \pgfkeysvalueof{/tikz/grid with coordinates/major step}}
        \foreach \x in {\grd@xa,\grd@xc,...,\grd@xb}
        \node[anchor=north] at (\x,\grd@ya) {\pgfmathprintnumber{\x}};
        \foreach \y in {\grd@ya,\grd@yc,...,\grd@yb}
        \node[anchor=east] at (\grd@xa,\y) {\pgfmathprintnumber{\y}};
      }
    }
  },
  minor help lines/.style={
    help lines,
    step=\pgfkeysvalueof{/tikz/grid with coordinates/minor step}
  },
  major help lines/.style={
    help lines,
    line width=\pgfkeysvalueof{/tikz/grid with coordinates/major line width},
    step=\pgfkeysvalueof{/tikz/grid with coordinates/major step}
  },
  grid with coordinates/.cd,
  minor step/.initial=.2,
  major step/.initial=1,
  major line width/.initial=2pt,
}
\makeatother

\usepackage{thmtools}
\usepackage{thm-restate}
\usepackage{etoolbox}
\makeatletter
\def\problem@s{}
\newcounter{problems@cnt}

\newcommand{\allproblems}{\problem@s}
\makeatother

\nc{\opr}[1]{\operatorname{#1}}
\nc{\oU}{{\opr{U}}}
\nc{\oO}{{\opr{O}}}
\nc{\oD}{{\opr{D}}}
\newcommand{\braa}[1]{\langle\langle #1 |}
\newcommand{\kett}[1]{| #1 \rangle\rangle}

\newcommand{\update}[1]{\textcolor{black}{#1}}

\usepackage{qcircuit}
\nc{\yj}[1]{\textcolor{teal}{\textbf{[yj: #1]}}}
\nc{\sxz}[1]{\textcolor{blue}{\textbf{[sxz: #1]}}}
\nc{\KM}[1]{\textcolor{violet}{\textbf{[KM: #1]}}}
\nc{\CH}[1]{\textcolor{purple}{\textbf{[ch: #1]}}}
\nc{\YA}[1]{\textcolor{cyan}{\textbf{[YA: #1]}}}
\usepackage[normalem]{ulem}
\nc{\hknew}[1]{\textcolor{violet}{#1}}
\nc{\XW}[1]{\textcolor{magenta}{\textbf{[XW: #1]}}}

\begin{document}
\title{Hypothesis testing of symmetry in quantum dynamics}

\author{Yu-Ao Chen}
\thanks{Y.-A. Chen and C. Zhu contributed equally to this work.}
\affiliation{Thrust of Artificial Intelligence, Information Hub, The Hong Kong University of Science and Technology (Guangzhou), Guangdong 511453, China}
\author{Chenghong Zhu}
\thanks{Y.-A. Chen and C. Zhu contributed equally to this work.}
\affiliation{Thrust of Artificial Intelligence, Information Hub, The Hong Kong University of Science and Technology (Guangzhou), Guangdong 511453, China}
\author{Keming He}
\affiliation{Thrust of Artificial Intelligence, Information Hub, The Hong Kong University of Science and Technology (Guangzhou), Guangdong 511453, China}
\author{Yingjian Liu}
\affiliation{Thrust of Artificial Intelligence, Information Hub, The Hong Kong University of Science and Technology (Guangzhou), Guangdong 511453, China}
\author{Xin Wang}
\email{felixxinwang@hkust-gz.edu.cn}
\affiliation{Thrust of Artificial Intelligence, Information Hub, The Hong Kong University of Science and Technology (Guangzhou), Guangdong 511453, China}

\date{\today}

\begin{abstract}
Symmetry plays a crucial role in quantum physics, dictating the behavior and dynamics of physical systems. In this paper, we develop a hypothesis-testing framework for quantum dynamics symmetry using a limited number of queries to the unknown unitary operation and establish the quantum max-relative entropy lower bound for the type-II error. We construct optimal ancilla-free protocols that achieve optimal type-II error probability for testing time-reversal symmetry (T-symmetry) and diagonal symmetry (Z-symmetry) with limited queries. Contrasting with the advantages of indefinite causal order strategies in various quantum information processing tasks, we show that parallel, adaptive, and indefinite causal order strategies have equal power for our tasks. We establish optimal protocols for T-symmetry testing and Z-symmetry testing for 6 and 5 queries, respectively, from which we infer that the type-II error exhibits a decay rate of $\mathcal{O}(m^{-2})$ with respect to the number of queries $m$. This represents a significant improvement over the basic repetition protocols without using global entanglement, where the error decays at a slower rate of $\mathcal{O}(m^{-1})$.
\end{abstract}

\maketitle

\textbf{Introduction.---}
Symmetry is a fundamental concept in physics, playing a crucial role in the formulation and understanding of physical laws~\cite{gross1996role, fano1996symmetries}. According to Noether's theorem, each symmetry in a physical system corresponds to a conservation law~\cite{noether1918invariante} and the recognition of symmetries enables the streamlining of calculations by reducing the number of degrees of freedom tied to conserved quantities. Thus, symmetry offers a powerful framework for analyzing and simplifying complex quantum systems. It enhances the protection of quantum states against errors~\cite{gottesman2010introduction}, ensures secure keys for quantum cryptography~\cite{wootters1982single,bennett2014quantum}, and deepens the understanding of entanglement~\cite{eisert2006entanglementquantuminformationtheory}.

Building on this foundational role, the detection and determination of different symmetries within quantum dynamics are of particular significance, as symmetry enables the inference and prediction of the behavior of unknown quantum systems by revealing fundamental characteristics of the Hilbert space~\cite{lieb1961two, nielsen1981no}. Furthermore, recent research in algorithmic learning~\cite{huangQuantumAdvantageLearning2022, aharonovQuantumAlgorithmicMeasurement2022} has shown that identifying the symmetries of infinite unitary subgroups validates significant quantum advantages and enhances the learning capabilities of models.

However, distinguishing symmetries within quantum dynamics remains a challenging task due to the inherent complexity and often non-intuitive behavior of quantum systems, requiring a detailed understanding of the system's Hamiltonian and the corresponding quantum states~\cite{aharonovQuantumAlgorithmicMeasurement2022, huangQuantumAdvantageLearning2022, chen2022exponential}. This task aligns with the broader framework of quantum property testing, where one seeks to determine specific properties of quantum systems while minimizing resources and measurements~\cite{wang2011property, montanaro2013survey, buhrman2008quantum,Caro2024}. Although significant progress has been made in testing symmetries of finite unitary groups through algorithmic approaches~\cite{laborde2022quantum, laborde2024quantumalgorithmsrealizingsymmetric}, much work remains in extending these methods to more complex cases~\cite{hiai2009quantum}.

In this paper, we frame the challenging task of identifying symmetry subgroups of unitaries as a hypothesis-testing problem in quantum dynamics.
Specifically, we focus on two significant symmetries as illustrative examples: time-reversal symmetry (T-symmetry) and diagonal symmetry (Z-symmetry). The former originates from the representation of the orthogonal group and is closely related to the quantum resource theory of imaginarity~\cite{wuResourceTheoryImaginarity2021, hickey2018quantifying}. Although complex numbers are fundamental in both the theoretical and experimental aspects of quantum physics~\cite{renou2021quantum, chen2022ruling, li2022testing, haug2023pseudorandom}, determining whether a system undergoes evolution of a time-reversal Hamiltonian remains essential. 
The latter, \update{Z-symmetry}, referring to diagonal Hamiltonian or unitaries, is regarded as less resourceful in experimental physics and quantum computing. From the perspective of quantum dynamics, diagonal unitary circuits elucidate the coherence in quantum computing~\cite{shepherd2009temporally, bremner2011classical}, efficiently generate quantum randomness~\cite{nakata2013diagonal, nakata2014generating, nakata2017efficient}, simulate classical thermodynamics~\cite{nakata2014diagonal}, and simplify circuit synthesis~\cite{zhang2022automatic}.
These symmetries, though fundamental, associate with null sets within the general unitary group, making their identification a non-trivial challenge.

\begin{figure}[t]
    \centering
    \includegraphics[width=\linewidth]{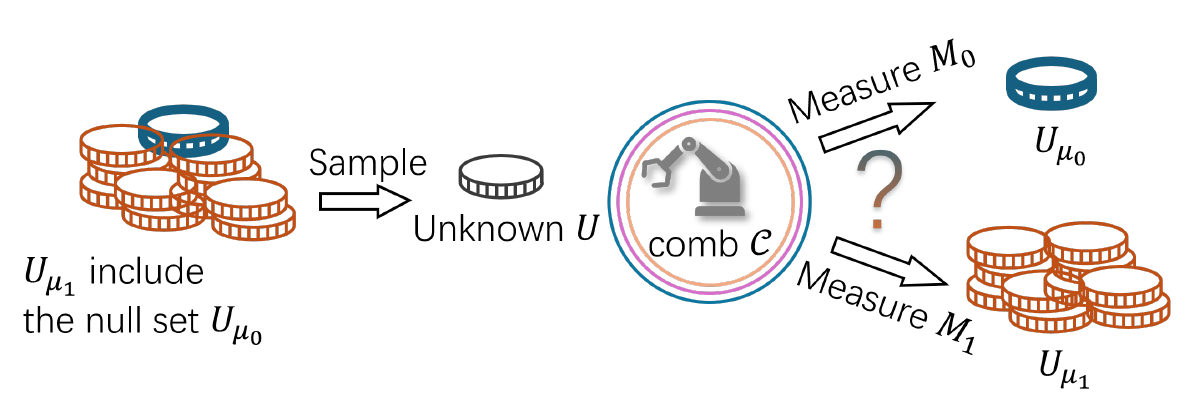}
    \caption{The framework of symmetry distinguishing task via quantum comb. Given two unitary distributions $\mu_0$, $\mu_1$ of interest, the task is to perform operations on the given unitaries and decide which distribution it belongs.}
    \label{fig:enter-label}
\end{figure}

We towards answer a basic question of what the optimal discrimination rate is between the null subset and the unitary under the two symmetries. For both T-symmetry and Z-symmetry testing, we focus on minimizing the type-II error probability while the unitary is allowed to be queried a limited number of times. We construct explicit and optimal protocols achieving the minimal type-II error when the number of unitary queries ranges from 1 to 6. \update{We further infer that the type-II error decays as $\cO(m^{-2})$ with the number of queries $m$ based on these optimal protocols, while the type-II error decays as $\cO(m^{-1})$ for the naive repetition strategies.} Our results also demonstrate that sequential and other causal strategies do not exhibit an advantage over these protocols, thereby establishing a counter-example to previous research in channel discrimination~\cite{Bavaresco_2021}, communication~\cite{chiribella2021indefinite, hayashi2021secure}, and quantum metrology~\cite{zhao2020quantum, chapeau2021noisy}.

\textbf{Hypothesis testing of unitary distributions.---}Given two $d$-dimensional unitary distributions $\mu_0$, $\mu_1$ of interest, the task is to perform operations on the given unitaries and decide which distribution it belongs. The task is similar to the binary-classification tasks in machine learning and it can also be treated as the channel-version of the composite hypothesis-testing task~\cite{berta2021composite}.

In this context, we could consider the distribution $\mu_1$ to be the Haar measure $\mu_U$ of $d$-dimensional unitary group $\oU(d)$, while the other distribution $\mu_0$ to be the Haar measure of some subgroup of unitary, denoted as $\oU_S(d)$, exhibiting a specific symmetry of interest, such as $\mu_\oO$ for \update{T-symmetry} unitaries or $\mu_\oD$ for diagonal unitaries. \update{T-symmetry} refers to the condition $e^{iHt} = e^{-iHt}$, which implies that the unitary $U$ is real, i.e. $U=U^{*}$, thereby classifying the corresponding operator $O$ as belonging to the orthogonal group $\opr{O}(d)$. Diagonal unitaries form a group isomorphism $\oU(1)^d$, referring to the matrices of the form $U_{D} = \sum_j e^{i\phi_j}\ketbra{j}{j}$ with a set of orthonormal basis $\{\ket{j}\}$, where we choose the computational basis without loss of generality. While T-symmetry and Z-symmetry are more implementable in experiments, it is impossible to sample them from the unitary group since they are the zero measure subset of $\opr{U}(d)$. For our application cases, we could assume the two distributions $\mu_0$ and $\mu_1$ disjointed by subtracting different distributions, which are approximated to the Haar measure.

For unknown quantum dynamics $\Tilde{U}$ sampled from the two distributions, one aims to determine which of the two it actually belongs to. A protocol (in the way of comb/super-channel) $\cC$ is implemented for the discrimination as $\cC(\Tilde{U})$. Then a quantum measurement, or positive operator-valued measure (POVM) $\{M_0, M_1\}$, where $M_1 = I - M_0$, is executed on the output of the desired protocol. If the measurement result outcomes $0$ from $M_0$, our decision goes to the quantum dynamics in $\mu_0$; if it outcomes $1$ from $M_1$, then the quantum dynamics is guessed from $\mu_1$. We then define the average error probability over the unitary distributions. In this task of \textit{hypothesis testing of unitary distributions}, there are two types of errors that can occur: type-I average error probability (false positive)
\begin{equation}
   \alpha(\cC) := \mathbb{E}_{U\sim\mu_0} \tr[\cC(U) M_1],
\end{equation}
and type-II average error probability 
\begin{equation}
   \beta(\cC) := \mathbb{E}_{U\sim\mu_1} \tr[\cC(U) M_0],
\end{equation}
where the expectation is taken through the unitary group.
Without the loss of generality, we specify $M_0$ and $M_1$ to be $\ketbra{0}{0}$ and $\ketbra{1}{1}$, respectively. Then the type-I error becomes
$\alpha(\cC) := \int_{U\sim\mu_0} \bra{1}\cC(U)\ket{1}$,
and the type-II error becomes $\beta(\cC) := \int_{U\sim\mu_1} \bra{0}\cC(U)\ket{0}$.

Given that $\mu_0$ is zero measure subset of $\mu_U$, it follows that one can always guess the unitary from $U(d)$ with a higher probability of success. It is then natural to consider asymmetric hypothesis testing, which asks about how small one of the errors can be subject to constraints on the other error. Specifically, it is concerned with the study of the optimized type-II error probability,
\begin{equation}
    \beta_{\epsilon}^{(m)} = \min_{\cC} \{\beta(\cC) ~|~ \alpha(\cC) \leq \epsilon , \cC \in \text{Comb}\},
\end{equation}
where $m$ denotes the number of queries of target unitaries, and the comb represents different causal strategies such as parallel, adaptive and indefinite protocols. It is reasonable to expect that the type-II error may decrease with more queries of target unitaries.

\textbf{Main Results.---} Following the best practice of quantum resource theory, the problem can be naturally formulated to find the optimal type-II error probability by querying the fixed number of unitaries $m$ with the following Semidefinite Programming (SDP),
\begin{equation}\label{eq:sdp_comb}
\begin{aligned}
  \beta^{(m)}_{\epsilon} =  \min \;\; &\tr\left[\cC\cdot(\Omega^{(m)}_{\mu_1}\ox\ketbra{0}{0})\right] \\
    \textit{s.t.} \;\; &\tr\left[\cC\cdot(\Omega^{(m)}_{\mu_0}\ox\ketbra{0}{0})\right] \geq 1 - \epsilon, \; \cC \in \text{Comb},
\end{aligned}
\end{equation}
where $\Omega_{(\cdot)}^{(m)}$ is defined as the \textit{performance operator} $\Omega_{(\cdot)}^{(m)}\coloneqq\mathbb E_{U\sim {(\cdot)}}\kett{U^{\otimes m}}\braa{U^{\otimes m}}$~\cite{quintino2022deterministic}, where $\kett{U^{\otimes m}}$ is the vectorization of $U$. The operator $\cC$ is a matrix representation of a quantum comb called the Choi matrix of a quantum comb, and it can be characterized by positivity and linear constraints~\cite{chiribella2008quantum}.
The solution of the SDP gives the optimal type-II error with the constrained type-I error tolerance. 
When $\epsilon \rightarrow 0$, namely, there is no type-I error. 

We would like to clarify that our setting is different from~\cite{montanaro2013survey, wang2011property, aharonovQuantumAlgorithmicMeasurement2022, chen2022exponential, haug2023pseudorandom} where certain type-I error is accepted to give efficient algorithms for identifying the desired set of unitaries. Given that the set of interest holds greater significance than other sets, we argue it is more natural to consider a general type-I error and explore the fundamental performance limits of the optimal protocols. \update{Moreover, previous work considers the worst-case scenario, aiming to ensure that the protocol performs reliably for the unitaries of interest even under the most challenging conditions. In contrast, the setting explored in this work is more aligned with an average-case perspective, where the goal is for the protocol to perform well on average across all unitaries of interest.}

 We next present the main theorem of this paper, an upper bound and a lower bound of hypothesis testing of quantum dynamics tasks. 

\begin{theorem}
    There is an upper bound $\beta_{up}^{(m)}$ that $\beta_{0}^{(m)}\le\beta_{up}^{(m)} $, where $\beta_{0}^{(m)}$ is the minimal average type-II error for testing the unitary distributions $\mu_0$ and $\mu_1$ within $m$ calls without type-I error:
\begin{equation}
   \begin{aligned} \beta_{up}^{(m)} =\min_{\ket\psi\in\cH_{d^m}} &\tr\left[\Omega^{(m)}_{\mu_1}\cdot(\ketbra{\psi}{\psi}^*\otimes\ketbra{\psi}{\psi})\right] \\
   \textit{ s.t. }& \tr\left[\Omega_{\mu_0}^{(m)}\cdot(\ketbra{\psi}{\psi}^*\otimes\ketbra{\psi}{\psi})\right]=1,
   \end{aligned} 
\end{equation}
or equivalently $\ket{\psi}$ is a common eigenstate of the $m$-th tensor powers of almost all unitaries in the support of $\mu_0$.

\label{thm:upper_bound_main}
\end{theorem}
We show Theorem 1 by constructing protocols satisfying the upper bound. More specifically, the protocol is given by \begin{theorem}
For a common eigenstate $\ket{\psi}$ of the $m$-th tensor powers of almost all unitaries in the support of $\mu_0$, denote a measure channel $\cM_\psi$ defined by its Choi operator
\begin{equation}
    \cJ_{\cM_\psi}\coloneqq\ketbra{\psi}{\psi}\otimes\ketbra{0}{0}+(\mathbbm{1}-\ketbra{\psi}{\psi})\otimes\ketbra{1}{1}.
\end{equation}
Then
\begin{equation}
    U\mapsto\cM_\psi(U^{\otimes m}\ket{\psi})
\end{equation}
is a protocol testing the unitary distributions $\mu_0$ and $\mu_1$ within $m$ calls without type-I error.
Moreover, its average type-II error is
\begin{equation}
\beta^{(m)}_0 = \tr\left[\Omega_{\mu_1}^{(m)}\cdot(\ketbra{\psi}{\psi}^*\otimes\ketbra{\psi}{\psi})\right].
\end{equation}
\end{theorem}
We give explicit constructions for orthogonal and diagonal unitaries discrimination and its upper bound of type-II error in appendix. 

\textbf{Quantum max-relative entropy bounds type-II error. ---} Intuitively, when a random unitary sampled from $\mu_1$ is very close to a unitary from the distribution $\mu_0$, the protocol faces difficulty in distinguishing between the two distributions. A well-known measure that quantifies the `distance' between two operators is the quantum max-relative entropy, which has been applied in various quantum information processing tasks~\cite{Datta2009, Bu2017, Wang2020, Zhu2024}. It is defined as $D_{\max}(P\|Q) := \inf\{\lambda : P \leq 2^\lambda Q \}$~\cite{Datta2009}, where $P$ and $Q$ are positive semidefinite operators. We then demonstrate that the lower bound of the average type-II error in distinguishability is constrained by the quantum max-relative entropy,
\begin{theorem}
    The lower bound for the minimal average type-II error for testing the distributions $\mu_0$ and $\mu_1$ within $m$ calls with type-I error $\epsilon$ is bounded by max-relative entropy,
\begin{equation}
    \beta^{(m)}_{\epsilon}\ge \beta^{(m)}_{\epsilon, low} = \frac{(1-\epsilon)}{2^{D_{\max}( \Omega_{\mu_0}^{(m)} \| \Omega_{\mu_1}^{(m)})}}.
\end{equation}
\end{theorem}
The bound quantifies the distance between the performance operator $\Omega_{\mu_0}^{(m)}$ and $\Omega_{\mu_1}^{(m)}$, indicating that a smaller distance makes it more challenging to distinguish between two distributions. Furthermore, this relationship provides another operational meaning to the quantum max-relative entropy in hypothesis testing of unitary distributions.

Notably, this lower bound of the minimal average type-II can be efficiently computed via SDP, which is given by,
\begin{equation}\label{eq:sdp_low}
    \beta_{\epsilon,low}^{(m)} =\max (1-\epsilon)t, \textit{ s.t. } \Omega_{\mu_1}^{(m)} - t\Omega_{\mu_0}^{(m)} \succeq 0,
\end{equation}
of which we remain the derivation in appendix. Note that we also introduce the necessary technique on link product~\cite{chiribella2009theoretical} and haar measure ~\cite{collins2006integration}  in appendix. It is noted that the optimization constraint in the dual formulation of this lower bound is relaxed by replacing $ C \succeq0$ in the original SDP with 
$\cC \in \text{Comb}$. This relaxation expands the optimization set beyond common causal strategies, such as sequential combs and indefinite causal order combs. These causal strategies are known to provide certain advantages in several quantum information processing tasks. However, we demonstrate that such advantages do not apply to our task below.

\textbf{Explicit Protocols.---} Next, we present several explicit ancilla-free parallel protocols designed to achieve the optimal type-II error probability for qubit cases. Specifically, we consider $\mu_0$ as the haar distribution and $\mu_1$ as the set of interest, such as T-symmetry and Z-symmetry sets. The overall protocol operates as follows: Given $m$ uses of the unitary $U$, it first prepares the input state $\ket{\psi_{\textrm{in}}}$, then processes the tensor product of the unitaries $U^{\otimes m}$ and finally performs a measurement using the same basis as $\ket{\psi_{\textrm{in}}}$. The overall process is illustrated in Fig.~\ref{fig:excplit_protocols}. Below, we provide detailed input states and measurement strategies for both the T-symmetry and Z-symmetry distribution distinguishability cases, aiming to achieve the optimal discrimination rate with a constrained type-I error of 0.

\begin{figure}
    \centering
    \includegraphics[width=1.0\linewidth]{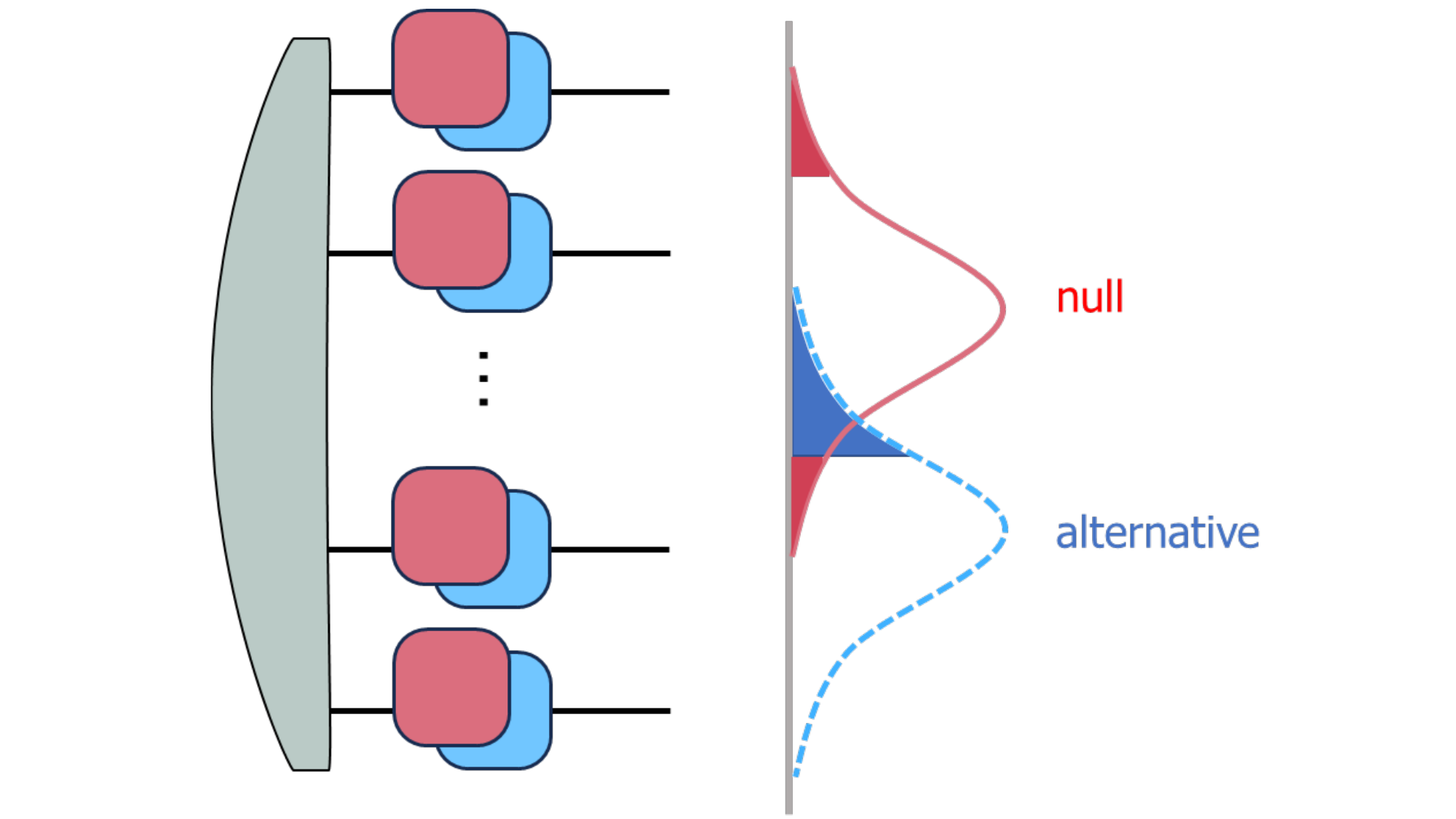}
    \caption{Illustration of the protocol: To query the unknown unitary $U$ $m$ times, the protocol for achieving the optimal type-II error probability is constructed using $U^{\otimes m}$ in parallel.}
    \label{fig:excplit_protocols}
\end{figure}

\textit{T-symmetry. } The two-slot parallel protocol uses the transpose trick~\cite{huangQuantumAdvantageLearning2022}, i.e. $(U\otimes U)\ket{\Phi^+} = (UU^T \otimes I)\ket{\Phi^+}$, where $\ket{\Phi^+}$ is the maximally entangled state. Leveraging the unique property of \update{T-symmetry} unitaries, one can achieve a type-II error probability $\beta^{(2)} = \frac{1}{3}$ with a corresponding type-I error probability $\alpha^{(2)} = 0$. This protocol appears to be optimal when the number of queries is fixed at 2. We then consider the performance of naively repeating this protocol $k$ times. \update{In this context, we also show that naive repetition protocols use 2-slot protocol $k$ times, resulting in an error probability $\frac{1}{2k+1}$.}

\renewcommand\arraystretch{1.5}
\begin{table}
\centering
\begin{tabular}{|c|c|c|}
\hline
\textbf{Queries $m$} & \textbf{Input State $\ket{\psi}$} & $\beta_{0}^{(m)}(\cC)$ \\ \hline
2 & $ \ket{\Phi^+} $ & $\frac{1}{3}$ \\ \hline
4 & \thead{$\frac{\sqrt5+1}{4}\ket{\Phi^+}_{01}\ket{\Phi^+}_{23}$ \\ $+ \frac{\sqrt5-1}{4}\ket{\Phi^+}_{02}\ket{\Phi^+}_{13}$} & $\frac{1}{6}$ \\ \hline
6 & \thead{$\frac{\sqrt7-2}{6}\ket{\Phi^+}_{05}\ket{\Phi^+}_{13}\ket{\Phi^+}_{24}$ \\ $+ \frac{\sqrt7+1}{6}\ket{\Phi^+}_{01}\ket{\Phi^+}_{24}\ket{\Phi^+}_{35}$ \\ $+\frac{\sqrt7+1}{6}\ket{\Phi^+}_{02}\ket{\Phi^+}_{13}\ket{\Phi^+}_{45}$} & $\frac{1}{10}$ \\ \hline
\end{tabular}
\caption{Summary of T-symmetry testing with different slots and corresponding results.}
\label{table:t_symmetry_input}
\end{table}

We further find that there are no improvements over the protocol by querying the unitary 3 times.  This is shown by finding the optimal protocol that only achieves the same result as the 2-slot solution. It is verified by the upper bound in Theorem~\ref{thm:upper_bound_main} that the common eigenstate is still limited. We also observe a similar phenomenon of querying the unitary 5 times, which shows no improvements over querying 4 times. However, there is an improvement over the protocol by querying the unitary 4 times, showing the type-II error probability of $\beta^{(4)}_0 = 1/6$. We summarize the protocols for the T-symmetry case in table~\ref{table:diagonal}, and we refer to detailed derivations in appendix.

\textit{Z-symmetry. } We construct explicit protocols demonstrating that when the number of queries $m$ is odd, the type-II error is $4/(m+1)(m+3)$, while the type-II error is $4/(m+2)^2$ for even $m$, for the number of queries ranges from 1 to 5. \update{Naive repetition protocols without using global entanglement, which use the one-slot protocol $k$ times, show that the type-II error probability is $1/(k+1)$}. This highlights an improvement achieved by querying the unknown unitaries in parallel. We summarize the protocols for diagonal case in table~\ref{table:diagonal} and we refer to detailed derivations in appendix.

\renewcommand\arraystretch{1.5}
\begin{table}
\centering
\begin{tabular}{|c|c|c|}
\hline
\textbf{Queries $m$} & \textbf{Input State $\ket{\psi}$} & $\beta_{0}^{(m)}(\cC)$ \\ \hline
1 & $\ket{0}$ & $\frac{1}{2}$ \\ \hline
2 & $ \sin\frac{\pi}{12}\ket{01} + \cos\frac{\pi}{12}\ket{10}$ & $\frac{1}{4}$ \\ \hline
3 & $ \frac{1}{\sqrt{2}}(\ket{001} + \ket{010})$ & $\frac{1}{6}$ \\ \hline
4 & \thead{$\frac{\sqrt{5}+1}{2\sqrt{6}}(\ket{0011} + \ket{0110})$ \\ $+\frac{\sqrt{5}-1}{2\sqrt{6}}(\ket{0101} + \ket{1001})$} & $\frac{1}{9}$ \\ \hline
5 & \thead{$\frac{\sqrt{5}+1}{8}(\ket{00011} + \ket{10100})$ \\ $+ \frac{\sqrt{5}-1}{8}(\ket{00101} + \ket{10010})$ \\ $+ \frac{\sqrt{5}}{4}(\ket{00110} + \ket{10001})$} & $\frac{1}{12}$ \\ \hline
\end{tabular}
\caption{Summary of Z-symmetry testing with slots varied between 1 to 5 and its corresponding type-II error.}
\label{table:diagonal}
\end{table}

For both T-symmetry and Z-symmetry cases, we evaluate the correctness of these protocols by \update{solving SDPs Eqs.\eqref{eq:sdp_comb} and \eqref{eq:sdp_low}} in MATLAB~\cite{MATLAB} using the interpreter CVX~\cite{cvx} with the solvers SDPT3~\cite{toh1999sdpt3} and obtain the upper and lower bound values for the number of slots 1-4. The simplified calculations for 1-6 slots can also be done with Mathematica~\cite{Mathematica}. We further provide the analytical solution to these bounds in appendix.

\textbf{Error tolerance.---} We have demonstrated that the constructed protocols are optimal for testing T-symmetry and Z-symmetry unitaries with zero type-I error. A natural question then arises: are these protocols also optimal when allowing a type-I error within a tolerance $\epsilon$? In appendix, we establish that the upper and lower bounds match, indicating that these optimal protocols for the zero type-I error case remain optimal when a type-I error is permitted. Specifically, the type-II error probability is given by $\beta_{\epsilon}^{(m)} = (1-\epsilon)\beta_0^{(m)}$.

\textbf{Concluding remarks.---} 
In this study, we have developed a hypothesis-testing framework that emphasizes minimizing the type-II error probability in quantum dynamics testing. We show that the upper bound for the minimal average type II error can be achieved using a common eigenstate of the tensor powers of all unitaries in support of $U_S$ of interest. On the lower bound side, we relax the quantum comb to include only the positivity constraint, resulting in a larger optimization set compared to several causal order strategies. We demonstrate that these causal order strategies do not offer any advantage in this task, providing a counterpoint to the previous results in the current literature. This phenomenon may be closely related to group symmetry~\update{\cite{kato2021single, owari2017single, hayashi2020quantum,chiribella2005optimal}}.

We also constructed optimal ancilla-free protocols designed to minimize the type-II error probability when querying the unitaries within limited times. \update{We demonstrate that protocols leveraging global entanglement achieve an optimal type-II error decay of $\mathcal{O} (m^{-2})$ with respect to the number of queries $m$. Furthermore, we show that it outperforms $m$ uses of repetition protocols in both the Z-symmetric and T-symmetric cases, where the optimal type-II error decay is $\mathcal{O}(m^{-1})$.} 
\update{Based on these findings, we conjecture that there exist testing protocols for quantum dynamics where the type-II error in distinguishing T- or Z-symmetric unitaries from Haar-random unitaries scales optimally as $\cO(m^{-2})$.}

We emphasize that our focus is on the symmetry of quantum dynamics, rather than the symmetry of quantum states, as these represent distinct resources in quantum resource theory. Recent studies have demonstrated that random ensembles of symplectic states are indistinguishable from random Haar unitary states~\cite{west2024random}. Our work may also highlight a fundamental distinction between static and dynamic resources in the context of this task. 
In addition, it would be interesting to explore the testing of resources~\cite{Chitambar2018,Diaz2018a,Zhu2024a,Regula2021b,WWS19,Zhu2024b,Jing2025,Seddon2019,Regula2021} for quantum operations.

\begin{acknowledgements}
\textit{Acknowledgments.---} We acknowledge the stimulating discussions with Masahito Hayashi. We thank Yin Mo, Chengkai Zhu, Xuanqiang Zhao and Hongshun Yao for helpful comments.
This work was partially supported by the National Key R\&D Program of China (Grant No.~2024YFE0102500), the National Natural Science Foundation of China (Grant. No.~12447107), the Guangdong Provincial Quantum Science Strategic Initiative (Grant No.~GDZX2403008, GDZX2403001), the Guangdong Provincial Key Lab of Integrated Communication, Sensing and Computation for Ubiquitous Internet of Things (Grant No.~2023B1212010007), the Quantum Science Center of Guangdong-Hong Kong-Macao Greater Bay Area, and the Education Bureau of Guangzhou Municipality.
\end{acknowledgements}

\widetext

\begin{center}
\textbf{\large Appendix for \\``Hypothesis testing of symmetry in quantum dynamics''}
\end{center}

\renewcommand{\theproposition}{S\arabic{proposition}}
\setcounter{proposition}{0}
\renewcommand{\thedefinition}{S\arabic{definition}}
\setcounter{definition}{0}

\renewcommand{\thefigure}{S\arabic{figure}}
\setcounter{figure}{0}
\renewcommand{\theequation}{S\arabic{equation}}
\setcounter{equation}{0}
\renewcommand{\thesection}{\Roman{section}}
\setcounter{section}{0}
\setcounter{secnumdepth}{4}

In this supplemental material, we first introduce some preliminaries and elaborate on the basic setup of this work. Next, we provide the rigorous formulation and proof of the theorem mentioned in the main text.

\section{Preliminaries and basic setup}

In this section, we elaborate on the basic setup of our work. For clarity, we first introduce some preliminaries and background. Throughout this paper, we consider a $N$-qubit system $\mathcal{H}^{\otimes N}$ with local Hilbert space being the two-dimensional complex linear space $\mathcal{H}=\mathbb{C}^2=\opr{span}\{\ket{0},\ket{1}\}$.

\subsection{Quantum comb and the Choi isomorphism}

Any quantum channel $\cN: \cL(\cI) \to \cL(\cO)$ can be represented by the Choi operator $\cN \in \cL(\cI \otimes \cO)$
\begin{equation}
    \cJ_{\cN} := \sum_{i,j}\ketbra{i}{j}_{\cI} \otimes \cN(\ketbra{i}{j})_{\cO},
\end{equation}
where $\{\ket{i}\}$ is the computational basis of $\cI$. Equivalently, it can be represented by vectorization of Kraus operators $\{K_i\}$, 
\begin{equation}
     \cJ_{\cN} = \sum_i \kett{K_i}\braa{K_i}_{\cI \cO},
\end{equation}
where the vectorization is 
\begin{equation}
    \kett{K_i} = \sum_j \ket{j}_\cI \otimes (K_i\ket{j})_{\cO}.
\end{equation}
The composition of two quantum channel $\cN: \cL(I) \to \cL(\cO_1)$ and $\cM: \cL(\cO_1) \to \cL(\cO_2)$ can be represented as,
\begin{equation}
    \cJ_{\cM \circ \cN} = \cJ_{\cM} * \cJ_{\cN},
\end{equation}
where $*$ is the link product~\cite{chiribella2009theoretical} defined by,
\begin{equation}
    A * B := \tr_{\cB}[(A^{T_{\cB}} \otimes I_{\cC})(I_\cA \otimes B)]
\end{equation}
for $A \in \cL(\cA \otimes \cB)$ and $B \in \cL(\cB \otimes \cC)$ and $A^{T_{\cB}}$ is the partial transpose of $A$ with respect to the subsystem $\cB$. For $C \in \cL(\cC \otimes \cD)$, The link product satisfies the commutativity and associativity relations as
\begin{equation}
\begin{aligned}
    A * B &= B * A \\
    (A * B) * C &= A * (B * C)
\end{aligned}
\end{equation}
If two operators $A$ and $B$ do not have an overlap system, the link product of $A$ and $B$ becomes the tensor product, i.e. $A * B = A \otimes B$.

\subsection{Haar unitary integrals}
The Haar measure distribution on the unitary group is defined as
\begin{equation}
    \mathbb E_{U\sim \mu} f(U):= \int_{U\sim \mu} f(U)d\mu(U)
\end{equation}
We consider the twirling operation of $k$-moment operator 
\begin{equation}
    M^{(k)}_{\mu}(\cdot)= \mathbb E_{U\sim \mu}U^{\dagger\otimes k}(\cdot)U^{\otimes k}
\end{equation}
For unitary group, the $k$-moment operator~\cite{collins2006integration} is calculated by
\begin{equation}
    \mathbb E_{U\sim \mu}U^{\dagger\otimes k}(\cdot)U^{\otimes k} = \sum_{\pi, \sigma \in S_{k}}\text{Wg}(\pi^{-1}, d)\tr\left(V^\dagger_d(\sigma )(\cdot) \right)V_d(\pi)
\end{equation}
where $S_k$ is the symmetric group, and $\text{Wg}(\pi^{-1}, d)$ is the Weingarten coefficients in terms of characters of the symmetric group, and $V_d(\pi)$ is the permutation operator given an element $\pi$ of the symmetric group. 

The 1-moment and 2-moment operators are simply expressed as
\begin{equation}
\begin{aligned}
    \mathbb E_{U\sim \mu}U^{\dagger}(\cdot)U =& \frac{\tr (\cdot)I}{d}
\end{aligned}
\end{equation}
and
\begin{equation}\label{eq:twirling_2st_order}
\begin{aligned}
    \mathbb E_{U\sim \mu}U^{\dagger\otimes 2}(\cdot)U^{\otimes 2} =& \frac{1}{d^{2}-1}\left[\tr(\cdot) I + \tr(S \cdot) S \right]
    - \frac{1}{d(d^{2}-1)}\left[\tr(\cdot) S + \tr(S \cdot) I \right],
\end{aligned}
\end{equation}

\subsection{Identifying between symmetries of unitary}
Our task is to identify matrices with symmetry from general unitary matrices using hypothesis testing. We denote two unitary distributions $\mu_0, \mu_1$, in which $\mu_0$ is the Haar measure of $d$-dimensional unitary group $\oU(d)$, and $\mu_1$ is the measure of some infinite subgroup of unitary. For our interest, we consider two kinds of symmetries of unitaries: T-symmetry and Z-symmetry
\begin{definition}[T-symmetry]
    A unitary $U_O$ is called T-symmetric if it is a real matrix satisfying $U_O = U^*_O$, i.e., an orthogonal matrix. 
\end{definition}

\begin{definition}[Z-symmetry]
    A unitary $U_D$ is called Z-symmetric if it is in the form $U_{D} = \sum_j e^{i\phi_j}\ketbra{j}{j}$ with computational basis. 
\end{definition}
The group of orthogonal matrices forms a subgroup of unitary groups, and so is a diagonal unitary group.  We denote the Haar measure distribution of unitary group $\mu_U$, orthogonal group $\mu_O$, and diagonal unitary group $\mu_D$.

We define the average error probability through the groups. There are two types of errors that can occur: type-I average error probability (false positive)
\begin{equation}
   \alpha(\cC) := \int_{U\sim\mu_0} dU \tr[\cC(U) M_1],
\end{equation}
and type-II average error probability 
\begin{equation}
   \beta(\cC) := \int_{U\sim\mu_1} dU \tr[\cC(U) M_0],
\end{equation}
where the integral goes through the unitary group. 
Without the loss of generality, we choose $M_0$ and $M_1$ to be $\ketbra{\bm{0}}{\bm{0}}$ and $\ketbra{\bm{1}}{\bm{1}}$, respectively. The type-I error becomes
$\alpha(\cC) = \int_{U\sim\mu_0} dU \bra{1}\cC(U)\ket{1}$,
and the type-II error becomes $\beta(\cC) = \int_{U\sim\mu_1} dU \bra{0}\cC(U)\ket{0}$.

\begin{proposition}[Performance operator]
For any unitary distribution $\mu$, and any $m$-slot comb $\cC$ with its Choi operator $\cJ_\cC$, denote
\begin{equation}
    \Omega_\mu^{(m)}\coloneqq\mathbb E_{U\sim \mu}\kett{U^{\otimes m}}\braa{U^{\otimes m}},
\end{equation}
and then we have
\begin{equation}
    \mathbb E_{U\sim \mu}\cC(U^{\ox m})=\mathbb E_{U\sim \mu}\cJ_\cC*\kett{U^{\otimes m}}\braa{U^{\otimes m}}=\cJ_\cC*\Omega_\mu^{(m)},
\end{equation}
where $*$ denotes the link product on systems $\bm{\cI}$ and $\bm{\cO}$.
\end{proposition}
We thus express the type-I and type-II error probability in terms of the performance operator $\alpha(\cC) = \bra{1}\cJ_\cC*\Omega_{\mu_0}^{(m)}\ket{1}$ and $\beta(\cC) = \bra{0}\cJ_\cC*\Omega_{\mu_1}^{(m)}\ket{0}$.

\section{Upper bound and protocol constructing}
\subsection{Upper bound}
To distinguish unitary distributions $\mu_0$ and $\mu_1$ without type-I error, we find that any common eigenvectors of unitaries in the support of $\mu_0$ have nice properties. Based on this, we construct a kind of protocol which indicates an upper bound for average type-II error without type-I error. 
Firstly, we find
\begin{lemma}\label{thm:trE}
For any state $\ket{\psi}\in\cH_{d^m}$ and any unitary distribution $\mu$, we have
\begin{equation}
    \mathbb E_{U\sim \mu}\left|\bra{\psi}U^{\otimes m}\ket{\psi}\right|^2=\tr\left[\Omega_\mu^{(m)}\cdot(\ketbra{\psi}{\psi}^*\otimes\ketbra{\psi}{\psi})\right].
\end{equation}
\end{lemma}
\begin{proof}
It is checked that
\begin{align}
    &\tr\left[\Omega_\mu\cdot(\ketbra{\psi}{\psi}^*\otimes\ketbra{\psi}{\psi})\right]\\
    =&\mathbb E_{U\sim \mu}\tr\left[\kett{U^{\otimes m}}\braa{U^{\otimes m}}\cdot(\ketbra{\psi}{\psi}^*\otimes\ketbra{\psi}{\psi})\right]\\ 
    =&\mathbb E_{U\sim \mu}\left|(\bra{\psi}^*\otimes\bra{\psi})\cdot\kett{U^{\otimes m}}\right|^2\\   
    =&\mathbb E_{U\sim \mu}\left|(\bra{\psi}^*\otimes\bra{\psi})\cdot\sum_j\ket{j}\otimes U^{\otimes m}\ket{j}\right|^2\\
    =&\mathbb E_{U\sim \mu}\left|\sum_j\braket{\psi}{j}^*\cdot\bra{\psi}U^{\otimes m}\ket{j})\right|^2\\
    =&\mathbb E_{U\sim \mu}\left|\sum_j\bra{\psi}U^{\otimes m}\ket{j}\braket{j}{\psi}\right|^2\\
    =&\mathbb E_{U\sim \mu}\left|\bra{\psi}U^{\otimes m}\ket{\psi}\right|^2
\end{align}
\end{proof}
\begin{corollary}\label{cor:common_es_alpha}
$\ket{\psi}$ is a common eigenstate of the $m$-th tensor powers of almost all unitaries in the support of $\mu$, if and only if
\begin{equation}
    \tr\left[\Omega_\mu^{(m)}\cdot(\ketbra{\psi}{\psi}^*\otimes\ketbra{\psi}{\psi})\right]=\mathbb E_{U\sim \mu}\left|\bra{\psi}U^{\otimes m}\ket{\psi}\right|^2=1
\end{equation}
\end{corollary}

Considering $\ketbra{\psi}{\psi}^*\otimes\ketbra{\psi}{\psi}$ as a part of the Choi operator of some quantum comb, we could extend it into a parallel comb, whose protocol is following:

\begin{theorem}
For $\ket{\psi}$ a common eigenstate of the $m$-th tensor powers of almost all unitaries in the support of $\mu_0$, denote a measure channel $\cM_\psi$ defined by its Choi operator
\begin{equation}
    \cJ_{\cM_\psi}\coloneqq\ketbra{\psi}{\psi}\otimes\ketbra{0}{0}+(\mathbbm{1}-\ketbra{\psi}{\psi})\otimes\ketbra{1}{1}.
\end{equation}
Then
\begin{equation}
    \cC_\psi:\ U\mapsto\cM_\psi(U^{\otimes m}\ket{\psi})
\end{equation}
is a protocol testing unitary distributions $\mu_0$ and $\mu_1$ without type-I error.
Moreover, its average type-II error is
\begin{equation}
    \beta(\cC_\psi)=\tr\left[\Omega_{\mu_1}^{(m)}\cdot(\ketbra{\psi}{\psi}^*\otimes\ketbra{\psi}{\psi})\right].
\end{equation}
\end{theorem}
\begin{proof}
By Corollary \ref{cor:common_es_alpha} and Lemma \ref{thm:trE}, the average type-I error and the average type-II error are
\begin{align}
   &\alpha(\cC_\psi)=1-\mathbb E_{U\in\mu_0}\bra{0}\cM_\psi(U^{\otimes m}\ket{\psi})\ket{0}
   =1-\mathbb E_{U\in\mu_0}(|\bra{\psi}U^{\otimes m}\ket{\psi}|^2=0,\\
   &\beta(\cC_\psi)=\mathbb E_{U\in\mu_1}\bra{0}\cM_\psi(U^{\otimes m}\ket{\psi})\ket{0}
   =\mathbb E_{U\in\mu_1}(|\bra{\psi}U^{\otimes m}\ket{\psi}|^2=\tr\left[\Omega_{\mu_1}^{(m)}\cdot(\ketbra{\psi}{\psi}^*\otimes\ketbra{\psi}{\psi})\right].
\end{align}
\end{proof}

\begin{corollary}
    Here is an upper bound for the minimal average type-II error for testing the unitary distributions $\mu_0$ and $\mu_1$ within $m$ calls without type-I error:
\begin{equation}
    \beta_{up}^{(m)} =\min_{\ket\psi\in\cH_{d^m}} \tr\left[\Omega^{(m)}_{\mu_1}\cdot(\ketbra{\psi}{\psi}^*\otimes\ketbra{\psi}{\psi})\right] \textit{ s.t. }\tr\left[\Omega_{\mu_0}^{(m)}\cdot(\ketbra{\psi}{\psi}^*\otimes\ketbra{\psi}{\psi})\right]=1,
\end{equation}
or equivalently $\ket{\psi}$ is a common eigenstate of the $m$-th tensor powers of almost all unitaries in the support of $\mu_0$.
\end{corollary}

\subsection{The optimal testing protocols}\label{appendix:protocols}

Here we propose protocols for T-symmetry testing and Z-symmetry testing satisfying the upper bound.  

\begin{example}\label{exmp:Tsym}(T-symmetry testing)
For T-symmetry testing with $2$ slots, 
\begin{equation}
    \text{let }\ket\psi=\ket{\Phi^+} \text{ and we have } \tr\left[\Omega_{\mu_U}^{(2)}\cdot(\ketbra{\psi}{\psi}^*\otimes\ketbra{\psi}{\psi})\right]=\frac13;
\end{equation}
for $4$ slot, 
\begin{equation}
    \text{let }\ket\psi=\frac{\sqrt5+1}{4}\ket{\Phi^+}_{01}\ket{\Phi^+}_{23}+\frac{\sqrt5-1}{4}\ket{\Phi^+}_{02}\ket{\Phi^+}_{13} \text{ and we have } \tr\left[\Omega_{\mu_U}^{(4)}\cdot(\ketbra{\psi}{\psi}^*\otimes\ketbra{\psi}{\psi})\right]=\frac16;
\end{equation}
for $6$ slot, let
\begin{equation}
    \ket\psi=\frac{\sqrt7-2}{6}\ket{\Phi^+}_{05}\ket{\Phi^+}_{13}\ket{\Phi^+}_{24}+\frac{\sqrt7+1}{6}\ket{\Phi^+}_{01}\ket{\Phi^+}_{24}\ket{\Phi^+}_{35}+\frac{\sqrt7+1}{6}\ket{\Phi^+}_{02}\ket{\Phi^+}_{13}\ket{\Phi^+}_{45}
\end{equation}
and we have
\begin{equation}
    \tr\left[\Omega_{\mu_U}^{(6)}\cdot(\ketbra{\psi}{\psi}^*\otimes\ketbra{\psi}{\psi})\right]=\frac{1}{10}.
\end{equation}
\end{example}

\begin{example}\label{exmp:diag}(Z-symmetry testing)
For Z-symmetry testing with $1$ slot, 
\begin{equation}
    \text{let }\ket\psi=\ket{0} \text{ and we have } \tr\left[\Omega_{\mu_U}^{(1)}\cdot(\ketbra{\psi}{\psi}^*\otimes\ketbra{\psi}{\psi})\right]=\frac12;
\end{equation}
for $2$ slots, 
\begin{equation}
    \text{let }\ket\psi=\sin\frac{\pi}{12}\ket{01}+\cos\frac{\pi}{12}\ket{10} \text{ and we have } \tr\left[\Omega_{\mu_U}^{(2)}\cdot(\ketbra{\psi}{\psi}^*\otimes\ketbra{\psi}{\psi})\right]=\frac14;
\end{equation}
for $3$ slots,
\begin{equation}
    \text{let }\ket\psi=\frac{1}{\sqrt{2}}(\ket{001}+\ket{010}) \text{ and we have } \tr\left[\Omega_{\mu_U}^{(3)}\cdot(\ketbra{\psi}{\psi}^*\otimes\ketbra{\psi}{\psi})\right]=\frac16;
\end{equation}
for $4$ slots,
\begin{equation}
    \text{let }\ket\psi=\frac{\sqrt{5}+1}{2\sqrt{6}}(\ket{0011}+\ket{0110})+\frac{\sqrt{5}-1}{2\sqrt{6}}(\ket{0101}+\ket{1001}) \text{ and we have } \tr\left[\Omega_{\mu_U}^{(4)}\cdot(\ketbra{\psi}{\psi}^*\otimes\ketbra{\psi}{\psi})\right]=\frac19;
\end{equation}
for $5$ slots, let
\begin{equation}
    \ket\psi=\frac{\sqrt{5}+1}{8}(\ket{00011}+\ket{10100})+\frac{\sqrt{5}-1}{8}(\ket{00101}+\ket{10010})+\frac{\sqrt5}{4}(\ket{00110}+\ket{10001})
\end{equation}
and we have
\begin{equation}
    \tr\left[\Omega_{\mu_U}^{(5)}\cdot(\ketbra{\psi}{\psi}^*\otimes\ketbra{\psi}{\psi})\right]=\frac{1}{12};
\end{equation}
\end{example}

\subsection{Proofs of optimal testing protocols}
In this section, we will prove that the optimal testing protocols satisfy the upper bound. 

\begin{proposition}\label{prop:SU2_expansion}
    Expand an unitary in $SU(2)$ as  
\begin{equation}
    U=p_0I+ip_1X+ip_2Y+ip_3Z,
\end{equation}
we have
\begin{equation}
    U\sim\mu_{\oU}\iff (p_0,p_1,p_2,p_3)\sim\opr{Uni}(\opr{S}^3),
\end{equation}
where $\opr{S}^3$ denotes the $3$-dimensional unit sphere in $\mathbb R^4$.
Then we have
\begin{equation}
    \mathbb E_{U\sim\mu_{\oU}}\left|\bra{\psi}U^{\otimes m}\ket{\psi}\right|^{2k}
    =\mathbb E_{(p_0,p_1,p_2,p_3)\sim\opr{Uni}(\opr{S}^3)}\left|\bra{\psi}U^{\otimes m}\ket{\psi}\right|^{2k}.
\end{equation}
\end{proposition}

\begin{lemma}\label{lem:SU2_integral}
For any polynomial $f(q_0^2,q_1^2)\in\mathbb C[q_0^2,q_1^2]$, we have
\begin{equation}
\begin{aligned}
    \mathbb E_{(p_0,p_1,p_2,p_3)\sim\opr{Uni}(\opr{S}^3)}f(p_0^2+p_1^2,p_2^2+p_3^2) 
    =\mathbb E_{(q_0,q_1)\sim\opr{Uni}(\opr{S}^1)}f(q_0^2,q_1^2)\cdot\pi|q_0q_1|
    =\frac{1}{2}\int_0^{2\pi}f(\cos^2\theta,\sin^2\theta)|\sin\theta\cos\theta|d\theta.
\end{aligned}
\end{equation}
\end{lemma}
\begin{proof}
Introduce 
\begin{equation}
    (p_0,p_1,p_2,p_3)=(\cos\theta\cos\alpha,\cos\theta\sin\alpha,\sin\theta\cos\beta,\sin\theta\sin\beta),\theta\in[0,2\pi),\alpha,\beta\in[0,\pi],
\end{equation}
and we have
\begin{align}
    &\mathbb E_{(p_0,p_1,p_2,p_3)\sim\opr{Uni}(\opr{S}^3)}f(p_0^2+p_1^2,p_2^2+p_3^2)
    =\frac{1}{2\pi^2}\int_0^{2\pi} d\theta\int_0^\pi d\alpha\int_0^\pi d\beta f(\cos^2\theta,\sin^2\theta)\cdot\left|\det\left(\frac{\partial p_j}{\partial(r,\theta,\alpha,\beta)}\middle)\right|\right|_{r=1}\\
    =&\frac{1}{2\pi^2}\int_0^{2\pi} d\theta\int_0^\pi d\alpha\int_0^\pi d\beta f(\cos^2\theta,\sin^2\theta)\cdot|\sin\theta\cos\theta|
    =\frac{1}{2}\int_0^{2\pi}f(\cos^2\theta,\sin^2\theta)|\sin\theta\cos\theta|d\theta\\
    =&\frac{1}{2\pi}\int_0^{2\pi}f(q_0^2,q_1^2)\cdot\pi |q_0q_1|d\theta
    =\mathbb E_{(q_0,q_1)\sim\opr{Uni}(\opr{S}^1)}f(q_0^2,q_1^2)\cdot\pi|q_0q_1|
\end{align}
\end{proof}

Here we prove the T-symmetry testing with 6 slots, and Z-symmetry testing with 5 slots using \ref{prop:SU2_expansion} and \ref{lem:SU2_integral}. The proof of protocols with other slots is analogous. 

\begin{proof}(Proof for Example \ref{exmp:Tsym})
For T-symmetry testing with $6$ slot, let
\begin{equation}
    \ket\psi=\frac{\sqrt7-2}{6}\ket{\Phi^+}_{05}\ket{\Phi^+}_{13}\ket{\Phi^+}_{24}+\frac{\sqrt7+1}{6}\ket{\Phi^+}_{01}\ket{\Phi^+}_{24}\ket{\Phi^+}_{35}+\frac{\sqrt7+1}{6}\ket{\Phi^+}_{02}\ket{\Phi^+}_{13}\ket{\Phi^+}_{45}
\end{equation}
and we have
\begin{align}
    &\tr\left[\Omega_{\mu_U}^{(6)}\cdot(\ketbra{\psi}{\psi}^*\otimes\ketbra{\psi}{\psi})\right]\\
    =&\mathbb E_{(p_0,p_1,p_2,p_3)\sim\opr{Uni}(\opr{S}^3)}\left|\left(p_0^2+p_2^2\right){}^3-6 \left(p_1^2+p_3^2\right) \left(p_0^2+p_2^2\right){}^2+6 \left(p_1^2+p_3^2\right){}^2 \left(p_0^2+p_2^2\right)-\left(p_1^2+p_3^2\right){}^3\right|^2\\
    =&\mathbb E_{(q_0,q_1)\sim\opr{Uni}(\opr{S}^1)}\left(q_0^6-6 q_1^2 q_0^4+6 q_1^4 q_0^2-q_1^6\right)^2\cdot\pi \left|q_0q_1\right|
    =\frac{1}{10}.
\end{align}
\end{proof}

\begin{proof}(Proof for Example \ref{exmp:diag})
For Z-symmetry testing with $5$ slots, let
\begin{equation}
    \ket\psi=\frac{\sqrt{5}+1}{8}(\ket{00011}+\ket{10100})+\frac{\sqrt{5}-1}{8}(\ket{00101}+\ket{10010})+\frac{\sqrt5}{4}(\ket{00110}+\ket{10001})
\end{equation}
and we have
\begin{align}
    &\tr\left[\Omega_{\mu_U}^{(5)}\cdot(\ketbra{\psi}{\psi}^*\otimes\ketbra{\psi}{\psi})\right]
    =\mathbb E_{(p_0,p_1,p_2,p_3)\sim\opr{Uni}(\opr{S}^3)}\left(p_0^2+p_3^2\right) \left(\left(p_0^2+p_3^2\right){}^2-3 \left(p_0^2+p_3^2\right) \left(p_1^2+p_2^2\right)+\left(p_1^2+p_2^2\right){}^2\right){}^2\\
    =&\mathbb E_{(q_0,q_1)\sim\opr{Uni}(\opr{S}^1)}q_0^2 \left(q_0^4-3 q_1^2 q_0^2+q_1^4\right){}^2\cdot\pi \left|q_0q_1\right|
    =\frac{1}{12}.
\end{align}
\end{proof}

\subsection{The native repetition testing protocols}
Beyond the optimal testing protocols, we could also repeat using a $m_1$-slot protocol and a $m_2$-slot protocol in parallel to construct a $(m_1+m_2)$-slot protocol. Clearly, this protocol could be hardly optimal. The following proposition gives a construction for native repetition protocols.
\begin{proposition}
Denote an $s$-qubit to single-qubit quantum channel $A$:
\begin{equation}
    A(\rho_0,\rho_1,\cdots,\rho_{s-1})=\ketbra{0}{0}-Z\cdot\prod_{j=0}^{s-1}\bra{1}\rho_j\ket{1}.
\end{equation}
Let $\ket{\psi_j}\in\cH_{d^{m_j}}$ be a common eigenstate of the $m_j$-th tensor powers of almost all unitaries in the support of $\mu_0$, then we have
\begin{equation}
    U\mapsto A(M_{\psi_0}(U^{\ox m_0}\ket{\psi_0}),M_{\psi_1}(U^{\ox m_1}\ket{\psi_1}),\cdots,M_{\psi_{s-1}}(U^{\ox m_{s-1}}\ket{\psi_{s-1}}))
\end{equation}
is a protocol testing unitary distributions $\mu_0$ and $\mu_1$ without type-I error, whose average type-II error is
\begin{equation}       
    \tr\left[\Omega_{\mu_1}^{(m)}\cdot(\ketbra{\psi}{\psi}^*\otimes\ketbra{\psi}{\psi})\right],
\text{ where }
    m=\sum_{j=0}^{s-1}m_j,\ 
    \ket{\psi}=\bigotimes_{j=0}^{s-1}\ket{\psi_j}.
\end{equation}
\end{proposition}

The following are several examples of naive repetition protocols. 

\begin{example}(T-symmetry testing with naive repetition protocols)
For T-symmetry testing with $4$ slots, apply the optimal $2$ slots protocol twice for the first two qubits and the last two qubits respectively
\begin{equation}
    \text{let }\ket\psi=\ket{\Phi^+}_{01}\ket{\Phi^+}_{23}\text{ and we have } \tr\left[\Omega_{\mu_U}^{(4)}\cdot(\ketbra{\psi}{\psi}^*\otimes\ketbra{\psi}{\psi})\right]=\frac15;
\end{equation}
for T-symmetry testing with $6$ slots, apply the optimal $2$ slots protocol three times
\begin{equation}
    \text{let }\ket\psi=\ket{\Phi^+}_{01}\ket{\Phi^+}_{23}\ket{\Phi^+}_{45}\text{ and we have } \tr\left[\Omega_{\mu_U}^{(6)}\cdot(\ketbra{\psi}{\psi}^*\otimes\ketbra{\psi}{\psi})\right]=\frac{1}{7};
\end{equation}
or apply both optimal $4$ slots protocol and optimal $2$ slots protocol once each
\begin{equation}
    \text{let }\ket\psi=\left(\frac{\sqrt5+1}{4}\ket{\Phi^+}_{01}\ket{\Phi^+}_{23}+\frac{\sqrt5-1}{4}\ket{\Phi^+}_{02}\ket{\Phi^+}_{13}\right)\ket{\Phi^+}_{45} \text{ and we have } \tr\left[\Omega_{\mu_U}^{(6)}\cdot(\ketbra{\psi}{\psi}^*\otimes\ketbra{\psi}{\psi})\right]=\frac{5}{42}.
\end{equation}

\end{example}

The following proposition discusses the naive repetition protocols of the 2-slot protocol, whose average type-II error tends to $\Theta(m^{-1})$ for the $m$-slot naive repeated protocol. Similarly, we could compute more on naive repeated protocols, but the average type-II errors may all be $\Omega(m^{-1})$.
\begin{proposition}
    For T-symmetry testing with $2k$ slot, apply the optimal 2 slots protocol $k$ times yielding
\begin{equation}
    \text{let }\ket\psi=\ket{\Phi^+}^{\ox k}\text{ and we have } \tr\left[\Omega_{\mu_U}^{(2k)}\cdot(\ketbra{\psi}{\psi}^*\otimes\ketbra{\psi}{\psi})\right]=\frac1{2k+1}.
\end{equation}
\end{proposition}
\begin{proof}
By Theorem \ref{thm:trE}, we have
\begin{equation}
    \tr\left[\Omega_{\mu_U}^{(2k)}\cdot(\ketbra{\psi}{\psi}^*\otimes\ketbra{\psi}{\psi})\right]=\mathbb E_{U\sim\mu_{\oU}}\left|\bra{\psi}U^{\otimes 2k}\ket{\psi}\right|^2=\mathbb E_{U\sim\mu_{\oU}}\left|\bra{\Phi^+}U^{\otimes 2}\ket{\Phi^+}\right|^{2k}.
\end{equation}
Using Proposition \ref{prop:SU2_expansion} and \ref{lem:SU2_integral} we have
\begin{equation}
\begin{aligned}
    \mathbb E_{U\sim\mu_{\oU}}\left|\bra{\Phi^+}U^{\otimes 2}\ket{\Phi^+}\right|^{2k}
    &=\mathbb E_{(p_0,p_1,p_2,p_3)\sim\opr{Uni}(\opr{S}^3)}\left|p_0^2-p_1^2+p_2^2-p_3^2\right|^{2k} \\
    &=\frac{1}{2\pi^2}\int_0^{2\pi} d\theta\int_0^\pi d\alpha\int_0^\pi d\beta \cos(2\theta)^{2k}\cdot\left|\det\left(\frac{\partial p_j}{\partial(r,\theta,\alpha,\beta)}\middle)\right|\right|_{r=1}\\
    =&\frac{1}{2\pi^2}\int_0^{2\pi} d\theta\int_0^\pi d\alpha\int_0^\pi d\beta \cos(2\theta)^{2k}\cdot\frac{|\sin(2\theta)|}{2} \\
    &=\frac{1}{4}\int_0^{2\pi} \cos(2\theta)^{2k}|\sin(2\theta)|d\theta
    =\int_0^{\pi/2} \cos(2\theta)^{2k}\sin(2\theta)d\theta\\
    =&\left.\frac{-\cos(2\theta)^{2k+1}}{2(2k+1)}\right|_{\theta=0}^{\pi/2}=\frac{1}{2k+1}.
\end{aligned}
\end{equation}
\end{proof}

\begin{example}(Z-symmetry testing with naive repeated protocols)
For Z-symmetry testing with $2$ slots, apply the optimal $1$ slot protocol twice for the first qubits and second qubits respectively
\begin{equation}
    \text{let }\ket\psi=\ket{00} \text{ and we have } \tr\left[\Omega_{\mu_U}^{(2)}\cdot(\ketbra{\psi}{\psi}^*\otimes\ketbra{\psi}{\psi})\right]=\frac13;
\end{equation}
for Z-symmetry testing with $3$ slot, apply the optimal $1$ slot protocol three times
\begin{equation}
    \text{let }\ket\psi=\ket{000} \text{ and we have } \tr\left[\Omega_{\mu_U}^{(3)}\cdot(\ketbra{\psi}{\psi}^*\otimes\ketbra{\psi}{\psi})\right]=\frac14;
\end{equation}
or apply both optimal 1 slots protocol and optimal 2 slots protocol once each
\begin{equation}
    \text{let }\ket\psi=\sin\frac{\pi}{12}\ket{001}+\cos\frac{\pi}{12}\ket{010} \text{ and we have } \tr\left[\Omega_{\mu_U}^{(3)}\cdot(\ketbra{\psi}{\psi}^*\otimes\ketbra{\psi}{\psi})\right]=\frac{3}{16};
\end{equation}
for Z-symmetry testing with $4$ slot, apply both optimal 1 slots protocol and optimal 3 slots protocol once each
\begin{equation}
    \text{let }\ket\psi=\frac{1}{\sqrt{2}}(\ket{0001}+\ket{0010}) \text{ and we have } \tr\left[\Omega_{\mu_U}^{(4)}\cdot(\ketbra{\psi}{\psi}^*\otimes\ketbra{\psi}{\psi})\right]=\frac{2}{15};
\end{equation}
apply both optimal 2 slots protocol twice
\begin{equation}
    \text{let }\ket\psi=\left(\sin\frac{\pi}{12}\ket{01}+\cos\frac{\pi}{12}\ket{10}\right)^{\otimes2} \text{ and we have } \tr\left[\Omega_{\mu_U}^{(4)}\cdot(\ketbra{\psi}{\psi}^*\otimes\ketbra{\psi}{\psi})\right]=\frac{11}{80};
\end{equation}
\end{example}

\section{Lower bound}
For any $m$-slot comb $C_{\mathbf{IO}F}$ without input system, denote
\begin{equation}
    M_{\mathbf{IO}}=\cJ_{C_{\mathbf{IO}F}}*_F\ketbra{0}{0}_F,
\end{equation}
the average type-II error and type-I error on unitary distribution $\mu$ testing are
\begin{align}
&\begin{aligned}
    &\mathbb E_{U\sim \mu_1}\bra{0}C_{\mathbf{IO}F}(U^{\otimes m})\ket{0}
    =\mathbb E_{U\sim \mu_1}\cJ_{C_{\mathbf{IO}F}}*_\mathbf{IO}\kett{U^{\otimes m}}\braa{U^{\otimes m}}_\mathbf{IO}*_F\ketbra{0}{0}_F
    =\tr\left[M\cdot\Omega_{\mu_1}^{(m)}\right],
\end{aligned}\\
    &\ \ \mathbb E_{U\sim \mu_0}\bra{1}C_{\mathbf{IO}F}(U^{\otimes m})\ket{1}=1-\mathbb E_{U\sim \mu_0}\bra{0}C_{\mathbf{IO}F}(U^{\otimes m})\ket{0}
    =1-\tr\left[M\cdot\Omega_{\mu_0}^{(m)}\right].
\end{align}

\begin{lemma}\label{lem:lower_bound_ep}
Here is a lower bound for the minimal average type-II error for testing the unitary distributions $\mu_0$ and $\mu_1$ within $m$ calls and type-I error at most $\epsilon$:
\begin{equation}
    \beta_{\epsilon,low}=\min\tr\left[\Omega_{\mu_1}^{(m)}\cdot C\right],\textit{ s.t. }\tr\left[\Omega_{\mu_0}^{(m)}\cdot C\right]\ge1-\epsilon,\ C \succeq0,
\end{equation}
or equally
\begin{equation}
    \beta_{\epsilon,low}=\max (1-\epsilon)t, \textit{ s.t. } \Omega_U^{(m)} - t\Omega_{\mu_1}^{(m)} \succeq 0. 
\end{equation}
\end{lemma}
\begin{proof}
When we ignore all conditions of $M_{\mathbf{IO}}$ except the semi-definite condition coming from quantum combs, we obtain a lower bound for the minimal average type-II error for testing the unitary distributions $\mu_0$ and $\mu_1$ within $m$ calls and type-I error at most $\epsilon$:
\begin{equation}
    \min\tr\left[\Omega_{\mu_1}^{(m)}\cdot C\right],\textit{ s.t. }\tr\left[\Omega_{\mu_0}^{(m)}\cdot C\right]\ge1-\epsilon,\ C \succeq0.
\end{equation}
Introducing the Lagrangian multipliers $t \in \mathbb{R}$, the Lagrangian of the primal SDP is given by
\begin{equation}
\begin{aligned}
    L(t) &= \tr\left[\Omega_{\mu_1}^{(m)}\cdot C\right] - t (\tr\left[\Omega_{\mu_0}^{(m)}\cdot C\right] - 1+\epsilon) \\
    &= \tr\left[(\Omega_{\mu_1}^{(m)} - t\Omega_{\mu_0}^{(m)})\cdot C\right] + (1-\epsilon)t.
\end{aligned}
\end{equation}
Since $C \succeq 0$, to keep the inner norm bounded it must hold that $\Omega_{\mu_1}^{(m)} - t\Omega_{\mu_0}^{(m)}\succeq 0$.  Then we have the dual SDP as follows,
\begin{equation}
    \max (1-\epsilon)t, \textit{ s.t. } \Omega_{\mu_1}^{(m)} - t\Omega_{\mu_0}^{(m)} \succeq 0,\ t\ge0. 
\end{equation}
Since $1-\epsilon\in[0,1]$, $\Omega_{\mu_1}^{(m)}\succeq0$, we find $t=0$ is an available solution. Thus we may omit the condition $t\ge0$. As a result, we have the equivalent SDP: 
\begin{equation}
    \max (1-\epsilon)t, \textit{ s.t. } \Omega_{\mu_1}^{(m)} - t\Omega_{\mu_0}^{(m)} \succeq 0.
\end{equation}
\end{proof}

\section{Unified optimality condition of symmetry testing}
In this section, we establish a linear relationship between the type-II error rates in symmetry testing: the type-II error without type-I error, denoted as $\beta^{(m)}_{0}$, and the type-II error with type-I error, denoted as $\beta^{(m)}_{\epsilon}$, is proportional.

\begin{proposition}
For
\begin{equation}
\begin{aligned}
  \beta^{(m)}_{\epsilon} =  \min \;\; &\tr\left[\cC\cdot(\Omega^{(m)}_{\mu_1}\ox\ketbra{0}{0})\right] \\
    \textit{s.t.} \;\; &\tr\left[\cC\cdot(\Omega^{(m)}_{\mu_0}\ox\ketbra{0}{0})\right] \geq 1 - \epsilon, \; \cC \in \text{Comb},
\end{aligned}
\end{equation}
we have
\begin{equation}
\begin{aligned}
  \beta^{(m)}_{\epsilon} \le (1 - \epsilon)\beta^{(m)}_{0}.
\end{aligned}
\end{equation}
\end{proposition}
\begin{proof}
Let $\cC_0$ be a comb with
\begin{equation}
\begin{aligned}
    &\tr\left[\cC_0\cdot(\Omega^{(m)}_{\mu_1}\ox\ketbra{0}{0})\right]=\beta^{(m)}_{0}\\
    &\tr\left[\cC_0\cdot(\Omega^{(m)}_{\mu_0}\ox\ketbra{0}{0})\right]=1,
\end{aligned}
\end{equation}
and $\cC_1$ be a comb with
\begin{equation}
    \forall\,U, \cC_1(U^{\ox m})=\ketbra{1}{1},
\end{equation}
which implies that
\begin{equation}
\begin{aligned}
    &\tr\left[\cC_1\cdot(\Omega^{(m)}_{\mu_1}\ox\ketbra{0}{0})\right]=0\\
    &\tr\left[\cC_1\cdot(\Omega^{(m)}_{\mu_0}\ox\ketbra{0}{0})\right]=0.
\end{aligned}
\end{equation}
Since the set of combs is always convex, we have
\begin{equation}
    \beta^{(m)}_{\epsilon}\le\beta^{(m)}((1-\epsilon)\cC_0+\epsilon\cC_1)=(1-\epsilon)\beta^{(m)}_0.
\end{equation}
\end{proof}

If the upper and lower bounds for both the type-I error free case and the type-I error inclusive case are proportional, and these bounds converge, then we conclude that the type-II error rates are proportional across these scenarios
\begin{proposition}
For
\begin{align}
\beta_{\epsilon,low}^{(m)}&\coloneqq(1-\epsilon)\beta_{0,low}^{(m)},\\
\beta_{\epsilon,up}^{(m)}&\coloneqq(1-\epsilon)\beta_{0,up}^{(m)},
\end{align}
We have
\begin{equation}
    \beta_{\epsilon,low}^{(m)}\le\beta^{(m)}_\epsilon\le\beta_{\epsilon,up}^{(m)}.
\end{equation}
\end{proposition}
\begin{corollary}
If $\beta_{low}^{(m)}=\beta_{up}^{(m)}$, then we have
$\beta_{\epsilon}^{(m)}=(1-\epsilon)\beta_{0}^{(m)}$.
\end{corollary}

\end{document}